\documentclass{article}
\usepackage{amssymb, amsmath, amsthm}

\usepackage[english]{babel} 
\usepackage[utf8]{inputenc} 
\usepackage{hyperref}
\usepackage[T1]{fontenc}
\usepackage{geometry}
\usepackage{cleveref}
\usepackage{breakcites}

\title{Expander Graphs are Non-Malleable Codes}
\author{Peter M.\ R.\ Rasmussen\thanks{University of Copenhagen and Basic Algorithms Research Copenhagen. \texttt{pmrr@di.ku.dk}.} \and Amit Sahai\thanks{UCLA. \texttt{sahai@cs.ucla.edu}}}

\date{\today}

\newcommand{\abs}[1]{\left\vert #1\right\vert}

\newcommand{\N}{\mathbb{N}}

\newcommand{\RR}{\mathcal R}
\newcommand{\M}{\mathcal M}

\newcommand{\F}{\mathcal F}

\newcommand{\LL}{\mathcal L}

\newcommand{\X}{\mathcal X}

\newcommand{\FF}{\mathbb F}

\newcommand{\ceil}[1]{\left\lceil #1 \right\rceil}

\DeclareMathOperator{\enc}{enc}
\DeclareMathOperator{\dec}{dec}
\DeclareMathOperator{\LD}{LD}

\newtheorem{definition}{Definition}
\newtheorem{notation}{Notation}
\newtheorem{theorem}[definition]{Theorem}
\newtheorem{proposition}[definition]{Proposition}

\begin{document}
	\maketitle

	\abstract{Any $d$-regular graph on $n$ vertices with spectral expansion $\lambda$ satisfying $n = \Omega(d^3\log(d)/\lambda)$ yields a $O\left(\frac{\lambda^{3/2}}{d}\right)$-non-malleable code for single-bit messages in the split-state model.}

	\section{Introduction}

	A split-state non-malleable code~\cite{NM} for single-bit messages consists of randomized encoding and decoding algorithms $(\enc, \dec)$. A message $m \in \{0,1\}$ is encoded as a pair of strings $(L,R) \in 
	\{0,1\}^k \times \{0,1\}^k$, such that $\dec(L,R) = m$. An adversary then specifies an arbitrary pair of functions $g,h: \{0,1\}^k \rightarrow \{0,1\}^k$. The code is said to be non-malleable if, intuitively, the message obtained as $\dec(g(L), h(R))$ is ``unrelated'' to the original message $m$. In particular, to be $\varepsilon$-non-malleable, it is enough~\cite{First} to guarantee that when the message $m$ is chosen uniformly at random and encoded into $(L,R)$,  the probability that $\dec(g(L), h(R)) = 1-m$ is at most $\frac12 + \varepsilon$. Since their introduction in 2010~\cite{NM}, split-state non-malleable codes have been the subject of intense study within theoretical computer science~\cite{NM,First,Aggarwal,ChaZuck,TamperedExtensions,Li}.

\vspace{1ex}
	In this work, we show that expander graphs immediately give rise to  split-state non-malleable codes for single-bit messages. 
	Specifically, we show that any $d$-regular graph on $n=2^k$ nodes with spectral expansion $\lambda$ satisfying $n = \Omega(d^3\log(d)/\lambda)$ yields a $O\left(\frac{\lambda^{3/2}}{d}\right)$-non-malleable code for single-bit messages in the split-state model. Our proof is elementary, requiring a little more than two (fullsize) pages to prove, 
	having at its heart two nested applications of the Expander Mixing Lemma.
	Furthermore, we only need expanders of high degree (e.g.,~$d = n^\varepsilon$), which can be constructed and analyzed easily (see, e.g.,~\cite{Luca-Blog} or the appendix), yielding $2^{-\Omega(k)}$-non-malleable codes.
	\paragraph{Comparison with Previous Work.}
	Until our work, all known proofs of security for explicit split-state non-malleable codes have required complex mathematical proofs, and all known such proofs either directly or indirectly used the mathematics behind constructions of two-source extractors~\cite{First,Aggarwal,ChaZuck,TamperedExtensions,Li}.  In fact, after constructing the first non-malleable code in the split-state model Dziembowski, Kazana, and Obremski wrote: ``This brings a natural question if we could show some relationship between the extractors and the non-malleable codes in the split-state model. Unfortunately, there is no obvious way of formalizing the conjecture that non-malleable codes need to be based on extractors''~\cite{First}. We thus simultaneously find the first simple, elementary solution to the problem of designing single-bit non-malleable codes (our proof being approximately one-third the length of the proof of security of the single-bit non-malleable code of~\cite{First}) and answer in the negative the implicit conjecture of~\cite{First}; it is not necessary to base constructions of non-malleable codes on the theory of extractors.
	
Our construction of non-malleable codes from expander graphs thus opens up a new line of attack in the study of split-state non-malleable codes. It is important to keep in mind that current constructions of non-malleable codes supporting messages of arbitrary length use many ideas pioneered in the construction of~\cite{First}, in particular the use of extractors. While we do not yet know how to generalize our results beyond single-bit messages, we speculate that further investigation building upon our work will reveal a deeper connection and more powerful simple constructions based on expanders.
	
	It should be noted that two-source extractors are well-known to exhibit expansion properties; however, in all previous proofs, much more than mere expansion was used to argue non-malleability. Indeed previous proofs apply extractors repeatedly; for instance the proof of~\cite{First} uses the extractor property several times (e.g., in equation (22) and using equation (43) in~\cite{First}). Previous proofs also highlight the nontriviality and care that is required in applying extractors correctly to yield a valid proof of non-malleability (e.g., the paragraph beginning with ``There are two problems with the above argument.'' found below equation (36) of~\cite{First}). With respect to the expansion properties of two-source extractors, it is not surprising that 1-bit non-malleable codes will have some sort of expansion properties. Our contribution is the converse: that  good expansion is \emph{sufficient} for the construction of non-malleable codes.
	
	\iffalse
	Points to touch upon:
	\begin{itemize}
	    \item Many different constructions (using the sparseness idea of \cite{First}) seem like they would be nm-codes. However, it has turned out to be really difficult to prove.
	    \item Greater perspective: Choosing a random graph is the same as choosing a good expander is the same as choosing a good (although perhaps not easily implementable) nm-code. Thus, one could view this as saying that in fact a random (unbalanced in the sense of sparseness) code is a good nm-code.
	    \item We need to handle that the result in the actual paper is not log or parameter optimal, but that the appendix takes care of this. Mention in introduction?
	\end{itemize}
	\fi
	
	\section{Preliminaries}
	We shall assume familiarity with the basics of codes and non-malleable codes. A cursory %introduction to the most relevant definitions %and intuition 
	review of relevant definitions
	can be found in the appendix.
	\begin{notation}[Graphs]
		A graph $G=(V, E)$ consists of vertices $V$ and edges $E\subset V\times V$. In this exposition every graph is undirected and $n=\abs V$ always denotes the number of vertices of the graph in question.
		\begin{itemize}
			\item For any $v\in V$ we denote by $N(v)$ the set of neighbors of $v$ in $G$.
		\item For any two subsets $S, T\subseteq V$ we denote by $E(S, T)$ the set of (directed) edges from $S$ to $T$ in $G$. I.e.\
			$E(S, T) = \{(v, u)\in S\times T\mid (v, u)\in E \}$.
		\end{itemize}
		\end{notation}
	\begin{definition}[Spectral Expander]
		Let $G=(V, E)$ be a $d$-regular graph, $A_G$ be its adjacency matrix, and $\lambda_1\geq \dots \geq \lambda_n$ be the eigenvalues of $A_G$. We say that $G$ is a $\lambda$ spectral expander if $\lambda \geq \max\{\abs{\lambda_2}, \dots, \abs{\lambda_n}\}$. 
	\end{definition}
		\begin{theorem}[Expander Mixing Lemma]\label{mixing}
		Suppose that $G=(V, E)$ is a $\lambda$ spectral expander. Then for every pair of subsets $S, T\subset V$ we have
		\begin{align*}
			\abs{\abs{E(S, T)}-\frac{d\cdot \abs S\cdot \abs T}{n}}\leq \lambda \sqrt{\abs S\cdot \abs T}.
		\end{align*}
	\end{theorem}
	Our results will rely on the following characterization of 1-bit non-malleable codes by Dziembowski, Kazana, and Obremski found in \cite{First}.
	\begin{theorem}\label{nflip_equal_nm}
		Let $(\enc, \dec)$ be a coding scheme with $\enc\colon \{0, 1\}\to \X$ and $\dec\colon \X\to\{0, 1\}$. Further, let $\F$ be a set of functions $f\colon \X\to \X$. Then $(\enc, \dec)$ is $\varepsilon$-non-malleable with respect to $\F$ if and only if for every $f\in \F$,
		\begin{align*} 
			\Pr_{b\xleftarrow u \{0,1\}}(\dec(f(\enc(b))) = 1-b) \leq \frac{1}{2}+\varepsilon,
		\end{align*}
		where the probability is over the uniform choice of $b$ and the randomness of $\enc$. 
	\end{theorem}

	\section{Results}
	We first formally introduce our candidate code and then prove that it is a non-malleable code.
	\subsection{Candidate Code}
	From a graph we can very naturally construct a coding scheme as follows.
	\begin{definition}[Graph Code]
		Let $G=(V, E)$ be a graph. The associated \emph{graph code}, $(\enc_G, \dec_G)$, consists of the functions
		\begin{align*}
			\enc_G&\colon \{0, 1\}\to V\times V,&
			\dec_G&\colon V\times V\to \{0, 1\}
		\end{align*}
		which are randomized and deterministic, respectively, and given by 
		\begin{align*}
			\enc_G(b) &= 
			\begin{cases}
				(u, v) \xleftarrow u (V\times V)\setminus E, & b=0,\\
				(u, v) \xleftarrow u E, & b=1,
			\end{cases}\\
			\dec_G(v_1, v_2) &= 
			\begin{cases}
				0, & (v_1, v_2)\not\in E,\\
				1, & (v_1, v_2)\in E.
			\end{cases}
		\end{align*}
	\end{definition}
	
	\subsection{Non-Malleability of Expander Graph Codes}
	Finally, arriving at the core of the matter, we first establish the following lemma casting the expression of Theorem \ref{nflip_equal_nm} in terms of graph properties.
	\begin{proposition}\label{ProbRepresentation}
		Let $G=(V, E)$ be  a graph, functions $g, h\colon V\to V$ be given, and $f=(g, h)\colon V\times V\to V \times V$ satisfy $f(u, v)=(g(u), h(v))$. For the probability that $f$ flips a random bit encoded by $\enc_G$, write 
		$$T = \Pr_{b\xleftarrow u\{0, 1\}}(\dec_G(f(\enc_G(b))) = 1-b)$$
		where the probability is taken over the randomness of $\enc_G$ and the sampling of $b$.
		Then
		\begin{align*}
			T &= \frac{1}{2}+\frac{1}{2d(n-d)} \sum_{(v, u)\in E}\left(\frac{d\abs{g^{-1}(v)}\cdot \abs{h^{-1}(u)}}{n}-\abs{E(g^{-1}(v), h^{-1}(u))}\right).
		\end{align*}
	\end{proposition}
	\begin{proof}
		For $b\in \{0, 1\}$ denote by $Q_b$ the probability
			$$Q_b = \Pr(\dec_G(f(\enc_G(b))) = 1-b)$$
		taken over the randomness of $\enc_G$. It is clear that $T=\frac{Q_0+Q_1}{2}$ and that by definition
		\begin{align*}
			Q_0 &= \Pr_{(v, u)\xleftarrow uV\times V\setminus E}\left[(g(v), h(u))\in E\right], &
			Q_1 &= \Pr_{(v, u)\xleftarrow u E}\left[(g(v), h(u))\not\in E\right].
		\end{align*}
		
		First, for $b=0$ we see that the number of non-edges that are mapped by $f$ to any given $(v, u)\in E$ is given by $\abs{g^{-1}(v)}\cdot \abs{h^{-1}(u)}-\abs{E(g^{-1}(v), h^{-1}(u))}$. There are $n(n-d)$ non-edges in $G$ so it follows that
		\begin{align*}
			Q_0 = \frac{\sum_{(v, u)\in E}\abs{g^{-1}(v)}\cdot \abs{h^{-1}(u)}-\abs{E(g^{-1}(v), h^{-1}(u))}}{n(n-d)}.
		\end{align*}
		
		Second, for $b=1$ the number of edges of $G$ that are mapped to non-edges by $f$ is given by $\sum_{(v, u)\not\in E}\abs{E(g^{-1}(v), h^{-1}(u))}$. Since there are $dn$ edges of $G$ to choose from when encoding the bit $b=1$, 
		\begin{align*}
			Q_1 = \frac{\sum_{(v, u)\not\in E}\abs{E(g^{-1}(v), h^{-1}(u))}}{dn}.
		\end{align*}
		Now, observing that the number of (directed) edges in the graph is $dn$ and 
		 that $\{g^{-1}(v)\}_{v\in V}$ and $\{h^{-1}(u)\}_{u\in V}$ are both partitions of $V$, we get
		\begin{align*}
			Q_1 &= \frac{dn-\sum_{(v, u)\in E}\abs{E(g^{-1}(v), h^{-1}(u))}}{dn} = 1-\frac{\sum_{(v, u)\in E}\abs{E(g^{-1}(v), h^{-1}(u))}}{dn}.
		\end{align*}
		\noindent
		Putting it all together,
		\begin{eqnarray*}
			T &=& \frac{\sum_{(v, u)\in E}\abs{g^{-1}(v)}\cdot \abs{h^{-1}(u)}-\abs{E(g^{-1}(v), h^{-1}(u))}}{2n(n-d)}  + \frac{1}{2}-\frac{\sum_{(v, u)\in E}\abs{E(g^{-1}(v), h^{-1}(u))}}{2dn}\\
			&=& \frac{1}{2} + \frac{1}{2d(n-d)} \sum_{(v, u)\in E}\left(\frac{d\abs{g^{-1}(v)}\cdot \abs{h^{-1}(u)}}{n}-\abs{E(g^{-1}(v), h^{-1}(u))}\right).
		\end{eqnarray*}
	\end{proof}
	\noindent
	We proceed immediately with the main theorem, which concludes the exposition. In order to keep this presentation short and to the point, more elaborate calculations, which save a few $\log$-factors, have been placed in the appendix as Theorem \ref{Thm-Elaborate}.
\begin{theorem}\label{mainTheorem}
		Let $G=(V, E)$ be $d$-regular with spectral expansion $\lambda$ satisfying $n = \Omega(d^3\log(d)^4/\lambda)$. Then $(\enc_G, \dec_G)$ is an $\tilde O\left(\frac{\lambda^{3/2}}{d}\right)$-non-malleable code in the split-state model.
	\end{theorem}
	\begin{proof}
		Let $f=(g, h)\colon V\times V\to V\times V$ be given. By Theorem \ref{nflip_equal_nm} and Proposition \ref{ProbRepresentation} we just need to show that
		\begin{align*}
			R = \frac{1}{2d(n-d)}\cdot \sum_{(v, u)\in E}\left(\frac{d\abs{g^{-1}(v)}\cdot \abs{h^{-1}(u)}}{n}-\abs{E(g^{-1}(v), h^{-1}(u))}\right)
		\end{align*}
		is bounded by $\tilde O\left( \frac{\lambda^{3/2}}{d} \right)$. Define the sets
		\begin{align*}
			G_1 &= \left\{v\in V\mid \abs {g^{-1}(v)}> \frac{n}{d^2}\right\}, & H_1 &=\left\{u\in V\mid \abs {h^{-1}(u)}> \frac{n}{d^2}\right\},\\
			G_2 &= \left\{v\in V\mid \abs {g^{-1}(v)}\leq \frac{n}{d^2} \right\}, & H_2 &=\left\{u\in V\mid \abs {h^{-1}(u)}\leq\frac{n}{d^2}\right\},
		\end{align*}
		for $i, j\in \{1, 2\}$ write
		\begin{align*}
			R_{i, j} = \frac{1}{2d(n-d)} \sum_{(v, u)\in E\cap (G_i\times H_j)}\left(\frac{d\abs{g^{-1}(v)}\cdot \abs{h^{-1}(u)}}{n}-\abs{E(g^{-1}(v), h^{-1}(u))}\right),
		\end{align*}
		and observe that $R = \sum_{1\leq i, j\leq 2}R_{i, j}$. 
		
		Consider the case when $i=2$. Simply bounding the terms of the form $\abs{g^{-1}(v)}\cdot \abs{h^{-1}(u)}$ by using that each vertex has only $d$ neighbours, we get
		\begin{align*}
			R_{2, 1}+R_{2, 2} &\leq \frac{1}{2n(n-d)} \sum_{(v, u)\in E\cap (G_2\times V)}\abs{g^{-1}(v)}\cdot \abs{h^{-1}(u)}\\
			&\leq \frac{1}{2n(n-d)}\cdot d \cdot \sum_{u\in V}\frac{n}{d^2}\cdot \abs{h^{-1}(u)}= \frac{n}{2(n-d)d}.\\
		\end{align*}
		Thus, 
			$R_{2, 1}+R_{2, 2} = O\left(d^{-1}\right)$.
		By symmetry, $R_{1, 2} = O\left(d^{-1}\right)$. It only remains to show that $R_{1,1}=\tilde O\left(\frac{\lambda^{3/2}}{d}\right)$. 
		To this end, partition $G_1$ and $H_1$, respectively, as
		\begin{align*}
			G_1^k &= \left\{ v\in G_1\mid \frac{n}{2^{k-1}}\geq \abs{g^{-1}(v)}>\frac{n}{2^{k}} \right\},&
			H_1^l &= \left\{ v\in H_1\mid \frac{n}{2^{l-1}}\geq \abs{h^{-1}(u)}>\frac{n}{2^{l}} \right\}
		\end{align*}
		for $1\leq k, l\leq \ceil{\log_2\left(d^2\right)}$. Now, focusing on each pair $G_1^k$ and $H_1^l$, we write
		\begin{align*}
			S_{k, l} = \frac{1}{2d(n-d)} \sum_{(v, u)\in E\cap (G_1^k\times H_1^l)}\left(\frac{d\abs{g^{-1}(v)}\cdot \abs{h^{-1}(u)}}{n}-\abs{E(g^{-1}(v), h^{-1}(u))}\right)
		\end{align*}
		and apply first the mixing lemma then the Cauchy-Schwartz inequality to get
		\begin{eqnarray*}
			2d(n-d)S_{k, l} %&= \sum_{v\in G_1^k}\left( \frac{d\abs{g^{-1}(v)}\cdot\abs{h^{-1}(N(v)\cap H_1^l)}}{n}-\abs{E\left( g^{-1}(v), h^{-1}(N(v)\cap H_1^l) \right)} \right)\\
			&=& \sum_{v\in G_1^k}\left( \frac{d\abs{g^{-1}(v)}\cdot\sum_{u\in N(v)\cap H_1^l}\abs{h^{-1}(u)}}{n}-\abs{E\left( g^{-1}(v), \bigcup_{u\in N(v)\cap H_1^l}h^{-1}(u) \right)} \right)\\
			&\leq& \sum_{v\in G_1^k}\lambda\sqrt{\abs{g^{-1}(v)}\cdot\sum_{u\in N(v)\cap H_1^l}\abs{h^{-1}(u)}}\\
			&\leq& \lambda \sqrt{\frac{n}{2^{k-1}}\cdot \frac{n}{2^{l-1}}}\cdot \sum_{v\in G_1^k}\sqrt{\abs{N(v)\cap H_1^l}}\\
			&\leq& 2\lambda n\cdot 2^{-\frac{l+k}{2}}\cdot\sqrt{\abs{G_1^k}}\cdot \sqrt{\abs{E(G_1^k, H_1^l)}}.
		\end{eqnarray*}
		We use the fact that $\abs{G_1^k}\leq 2^k, \abs{H_1^l}\leq 2^l$, apply the mixing lemma to the last factor, and wield Jensen's inequality on the arising square root to obtain
		\begin{align*}
			d(n-d)S_{k, l}&\leq \lambda n\cdot 2^{-\frac{l+k}{2}}\cdot \sqrt{\abs{G_1^k}}\cdot \sqrt{\frac{d\cdot\abs{G_1^k}\cdot \abs{H_1^l}}n+\lambda\sqrt{\abs{G_1^k}\cdot \abs{H_1^l}}}\\
			&\leq \lambda\sqrt{2^k dn}+2^{\frac{k-l}4}\lambda^{3/2}n \leq \lambda\cdot \sqrt{d^3n}+2^{\frac{k-l}4}\lambda^{3/2}n.
		\end{align*}
		By symmetry of $k$ and $l$, $d(n-d)S_{k, l} \leq \lambda\cdot \sqrt{d^3n}+2^{\frac{l-k}4}\lambda^{3/2}n$. Thus,
		\begin{align*}
			R_{1, 1}&=\sum_{1\leq k, l\leq \ceil{\log_2(d^2)}}S_{k, l}\\
			&\leq O\left(\frac{\lambda\log(d)^2\cdot\sqrt d}{\sqrt{n}}\right) + O\left(\frac{\lambda^{3/2}}{d}\right)\cdot \sum_{1\leq k, l\leq \ceil{\log_2(d^2)}}2^{-\frac{\abs{k-l}}4}\\
			&= O\left( \frac{\log(d)\lambda^{3/2}}{d} \right).
		\end{align*}
	\end{proof}
	
	    \section*{Acknowledgements}
        A significant effort was made to simplify our proof as much as possible, which eventually resulted in the approximately 2-page proof of our main result presented here; we thank Anders Aamand and Jakob B\ae k Tejs Knudsen for suggestions and insights regarding the main theorem that helped simplify and improve the results presented. Furthermore, we thank Aayush Jain, Yuval Ishai, and Dakshita Khurana for early discussions regarding simple constructions of split-state non-malleable codes not based on expander graphs.
        
        Research supported in part from a DARPA/ARL SAFEWARE award, NSF Frontier Award 1413955, and NSF grant 1619348, BSF grant 2012378, a Xerox Faculty Research Award, a Google Faculty Research Award, an equipment grant from Intel, and an Okawa Foundation Research Grant. This material is based upon work supported by the Defense Advanced Research Projects Agency through the ARL under Contract W911NF-15-C- 0205. The views expressed are those of the authors and do not reflect the official policy or position of the Department of Defense, the National Science Foundation, or the U.S. Government.
        
        Research supported in part by grant 16582, Basic Algorithms Research Copenhagen (BARC), from the VILLUM Foundation.
	
	\bibliographystyle{alpha}
	\bibliography{biblio}
	\newpage
	\appendix
	
	\section{Definitions for Split-State Non-Malleable Codes}
	
	Here, we recall the basic definition of a split-state
	non-malleable code due to~\cite{NM}.
	
	 \begin{definition}[Coding scheme]
	 	We define a \emph{coding scheme} to be a pair of functions $(\enc, \dec)$. The encoding function $\enc\colon \M\to \X$ is randomized while the decoding function $\dec\colon \X\to \M\cup \{\bot\}$ is deterministic. Further, for all $s\in \M$ the pair satisfies 
	 	\begin{align*}
	 		\Pr[\dec(\enc(s))=s] = 1
	 	\end{align*}
	 	where the probability is taken over the randomness of $\enc$.
	 \end{definition}

	\begin{definition}[Split State Non-Malleable Code]
	 	A coding scheme $(enc, \dec)$, $\enc\colon \M\to \LL\times\RR$ and $\dec\colon \LL\times\RR\to \M\cup \{\bot\}$,  is $\varepsilon$-non-malleable in the split state model if for every pair of functions $g\colon \LL\to\LL, h\colon \RR\to\RR$ and writing $f=(g, h)$ there exists a distribution $D_f$ supported on $\M\cup\{*, \bot\}$ such that for every $s\in \M$ the two random variables defined by the experiments
	 	\begin{align*}
	 		A_f^s &= \left\{ \substack{(L, R)\leftarrow \enc(s); \\ \text{Output } \dec(g(L), h(R)) } \right\}\\
	 		B_f^s &= \left\{ \substack{ \tilde s \leftarrow D_f; \\ \text{If } \tilde s = * \text{ output } s \text{ else output } \tilde s }\right\}
	 	\end{align*}
	 	%satisfy $A_f^s\approx_{\varepsilon} B_f^s$.
	 	have statistical distance at most $\varepsilon$.
	 \end{definition}

	\section{Deliver Us from Log Factors}
	A more thorough analysis of the sums in the proof of Theorem \ref{mainTheorem} allows us to get slightly better bounds. The technicalities are of little interest to the big picture and were hence omitted in the body of the paper. The addition consists of an alternative ending to the proof of Theorem \ref{mainTheorem}.
    \begin{theorem}\label{Thm-Elaborate}
            Let $G=(V, E)$ be $d$-regular with spectral expansion $\lambda$ satisfying $n = \Omega(d^3\log(d)/\lambda)$. Then $(\enc_G, \dec_G)$ is an $O\left(\frac{\lambda^{3/2}}{d}\right)$-non-malleable code in the split-state model.
    \end{theorem}
    \begin{proof}
         At the very end of the proof of Theorem \ref{mainTheorem}, we arrived at
        \begin{align*}
            d(n-d)S_{k, l}\leq 2^{-\frac{l+k}2}\lambda n\cdot \sqrt{\abs{G_1^k}}\cdot \sqrt{\frac{d\cdot \abs{G_1^k}\cdot \abs{H_1^l}}{n}+\lambda\cdot \sqrt{\abs{G_1^k}\cdot \abs{H_1^l}}}.
        \end{align*}
        Applying Jensen's inequality, we get
        \begin{align}\label{eq-SKL}
            S_{k, l}&\leq O\left(\frac{\lambda}{\sqrt{dn}}\right)\cdot 2^{-\frac{l+k}2}\cdot \abs{G_1^k}\cdot \sqrt{\abs{H_1^l}}+O\left( \frac{\lambda^{3/2}}{d} \right)\cdot 2^{-\frac{l+k}2}\cdot \sqrt[4]{\abs{G_1^k}^3\cdot \abs{H_1^l}}
        \end{align}
        with the functions hidden by the $O$-notation being independent of $k, l$.

        Now, note that 
        \begin{align}\label{gh-ineqs}
            \abs{g^{-1}(G_1^k)}&\geq \frac{n\cdot\abs{G_1^k}}{2^k} & \abs{h^{-1}(H_1^l)}&\geq \frac{n\cdot\abs{H_1^l}}{2^l}
        \end{align}
         and for all $k\leq \ceil{\log_2(d^2)}$ we have $\frac{\abs{G_1^k}}{2^{k/2}}\leq 2d$. We shall bound each of the terms of \eqref{eq-SKL} separately.
        
        First,
        write 
        \begin{align*}
            L = \sum_{1\leq k, l\leq \ceil{\log_2(d^2)}} \left(2^{-\frac{l+k}2}\cdot \abs{G_1^k}\cdot \sqrt{\abs{H_1^l}}\right).
        \end{align*}
        Using the Cauchy-Schwartz inequality in the second inequality,
        \begin{align*}
            L
            &\leq 2d\cdot \sum_{1\leq l\leq\ceil{\log_2(d^2)}}\sqrt{2^{-l}\abs{H_1^l}}\\
            &\leq O\left(d\cdot \sqrt{\log(d)}\right)\cdot \sqrt{ \sum_{1\leq l\leq\ceil{\log_2(d^2)}} 2^{-l}\cdot \abs{H_1^l}}\\
            &\leq O\left(d\cdot \sqrt{\log(d)}\right)\cdot \sqrt{ \sum_{1\leq l\leq\ceil{\log_2(d^2)}} \frac{\abs{h^{-1}(H_1^l)}}{n}}\\
            &=O\left(d\cdot \sqrt{\log(d)}\right)
        \end{align*}
        since the $H_1^l$ are disjoint subsets of $V$. In conclusion, 
        \begin{align*}
            O\left(\frac{\lambda}{\sqrt{dn}}\right)\cdot \sum_{1\leq k, l\leq \ceil{\log_2(d^2)}} 2^{-\frac{l+k}2}\cdot \abs{G_1^k}\cdot \sqrt{\abs{H_1^l}}
            &=O\left(\frac{\lambda \cdot\sqrt{d\cdot log(d)}}{\sqrt n}\right)\\ & = O\left(\frac{\lambda^{3/2}}d  \right).
        \end{align*}
        
        Second, let $k\leq l$ and write $t=l-k$. We now bound the sum using \eqref{gh-ineqs}. Write
        \begin{equation*}
            K = \sum_{1\leq k< l\leq \ceil{\log_2(d^2)}} 2^{-\frac{l+k}2}\cdot \sqrt[4]{\abs{G_1^k}^3\cdot \abs{H_1^l}}.
        \end{equation*}
        Then
        \begin{align*}
            K&\leq \sum_{1\leq k< l\leq \ceil{\log_2(d^2)}} \left(\frac{2^{\frac{k-l}4}}n\cdot \sqrt[4]{\abs{g^{-1}(G_1^k)}^3\cdot \abs{h^{-1}(H_1^l)}}\right)\\
            &\leq \sum_{t=0}^{\ceil{\log_2(d^2)}}\left(\frac{2^{-\frac{t}4}}n\sum_{l=t}^{\ceil{\log_2(d^2)}}\sqrt[4]{\abs{g^{-1}(G_1^{l-t})}^3\cdot \abs{h^{-1}(H_1^l)}}\right)\\
            &\leq \sum_{t=0}^{\ceil{\log_2(d^2)}}\left(\frac{2^{-\frac{t}4}}n\left(\sum_{l=t}^{\ceil{\log_2(d^2)}}\abs{g^{-1}(G_1^{l-t})}\right)^{3/4}\cdot \left(\sum_{l=t}^{\ceil{\log_2(d^2)}}\abs{h^{-1}(H_1^{l})}\right)^{1/4}\right)\\
            &\leq \sum_{t=0}^{\ceil{\log_2(d^2)}} 2^{-\frac{t}4} = O(1),
        \end{align*}
        where the third inequality is established using H\"older's inequalty.

        It now follows that 
        \begin{align*}
        		\sum_{1\leq k\leq l\leq \ceil{\log_2(d^2)}}S_{k, l} = O\left( \frac{\lambda^{3/2}}{d} \right).
        \end{align*}
        By symmetry of $k$ and $l$,
        \begin{align*}
            R_{1, 1} = \sum_{1\leq k, l\leq \ceil{\log_2(d^2)}}S_{k, l} = O\left( \frac{\lambda^{3/2}}{d} \right),
        \end{align*}	
        which completes the proof.
    \end{proof}

	\section{Instantiating Our Construction}
	Using our results to instantiate an efficient, secure split-state non-malleable code, we require a family of graphs $\{G_k\}_{k\in \N}$, where each $G_k=(V_k, E_k)$ is $d_k$-regular with spectral expansion $\lambda_k$, satisfying the following:
	\begin{enumerate}
		\item The function $\varepsilon(k) = \frac{\lambda_k^{3/2}}{d_k}$ is negligible.
		\item We have $n_k=\abs{V(G_k)} = \Omega(d_k^3\log(d_k)/\lambda_k)$
	    \item Both sampling an edge $(u, v)\xleftarrow u E_k$ and sampling a non-edge $(u, v)\xleftarrow u (V_k\times V_k)\setminus E_k$ can be done in time polynomial in $k$.
	    \item Determining membership of a pair $(u, v)\in V\times V$ in $E(G_k)$ can be done deterministically in time polynomial in $k$.
	\end{enumerate}
	Given such a family of graphs it is clear that the corresponding graph code $(\enc_{G_k}, \dec_{G_k})$ is an efficiently computable non-malleable code.
	\subsection{Instantiation with High-Degree Cayley Graphs}
        Explicit constructions of such families of graphs do indeed exist. We shall here give an example from~\cite{Luca-Blog} %https://lucatrevisan.wordpress.com/2011/02/28/cs359g-lecture-16-constructions-of-expanders/ or https://people.eecs.berkeley.edu/~luca/cs359g/lecture16.pdf
        from the class of graphs known as Cayley graphs. The construction is as follows.
        
        \begin{definition}
            For $p$ a prime and $1\leq t<p$ let the graph $\LD_{p, t}$ have vertex set $\FF_p^{t+1}$ and edge set
            \begin{align*}
                E(\LD_{p, t}) = \left\{(x, x+(b, ab, a^2b, \dots, a^tb)) \mid x\in \FF_p^{t+1}, a, b\in \FF_p\right\},
            \end{align*}
            i.e.\ $x, y\in V(\LD_{p, T})$ are connected by an edge if and only if there exists $a, b\in \FF_p$ such that $y = x+(b, ab, a^2b, \dots, a^tb)$.
        \end{definition}
        It is worth nothing that the graph $\LD_{p, t}$ is $\LD_{p, t}$ is $p^2$-regular and that it is undirected as $x$ is connected to $y$ if and only if $y$ is connected to $x$.
        \\

        Now, let $t=5$ and for each $k\in \N$ let $p_k$ be some $k$-bit prime. We consider the family of graphs $\{\LD_{p_k, 5}\}_{k\in \N}$ for our instantiation. In the following, we shall check the criteria from the beginning of the section point by point. 
        
        \begin{enumerate}
            \item The family of graphs $\LD_{p, t}$ has great expander properties.  
        \begin{theorem}[Trevisan \cite{Luca-Blog}]
            For $1<t<p$, the graph $\LD_{p, t}$ is a $pt$-spectral expander.
        \end{theorem}
        This fact allows us to note that for our particular choice of graphs, $\varepsilon(k) = \frac{(p_kt)^{3/2}}{p_k^2}<\frac{12}{\sqrt p_k}$, which in fact is $2^{-\Omega(k)}$ and the representation size is $O(k)$ bits. 
        \item We have $\Omega\left( \frac{d_k^3\log(d_k)}{\lambda_k} \right)=\Omega(p^5\log(p))$ such that indeed,
        $$n_k=\abs{V(\LD_{p_k, 5})} = p^6 = \Omega\left( \frac{d_k^3\log(d_k)}{\lambda_k} \right).$$
        \item Sampling an edge $(u, v)\xleftarrow u E(\LD_{p_k, t})$ is simply a question of picking $x\in \FF_{p_k}^{t+1}, a, b\in \FF_{p_k}$ uniformly at random and then outputting the edge $(x, x+(b, ab, a^2b, \dots, a^tb))$. 
        
        To pick a non-edge, simply sample two random vertices $u, v\in \FF_{p_k}^{t+1}$ uniformly at random and check (with the procedure to be specified below) whether $(u, v)\in E(\LD_{p_k, t})$. Since for $t>1$ the probability of hitting an edge with such a random choice is $\leq 1/{p_k}$, the expected number of repetitions is constant and hence the procedure takes expected polynomial time.
        \item To test membership of some $(u, v)\in \left(\FF_{p_k}^{t+1}\right)^2$ in $E(\LD_{p_k, t})$, perform the following operation: Compute $x=u-v$ and write $x=(x_0, \dots, x_t)$. It is now trival to check whether $\left(1, \frac{x_1}{x_0}, \dots, \frac{x_t}{x_0}\right)$ is of the form $(1, a, a^2, \dots, a^t)$.
        \end{enumerate}
        
\end{document}